\RequirePackage{xcolor}
\documentclass[journal,twoside,web]{ieeecolor} 
\usepackage{lcsys}
\usepackage[T1]{fontenc}
\usepackage[utf8]{inputenc}
\usepackage[english]{babel}

\usepackage{amsmath}

\usepackage{amsfonts} 
\usepackage{amsthm}
\usepackage{amssymb} 
\usepackage{cite} 
\usepackage{graphicx} 
\usepackage{xcolor}
\usepackage{enumerate}
\usepackage{verbatim}
\usepackage[hyphens]{url}
\PassOptionsToPackage{hyphens}{url}

\theoremstyle{plain}
\newtheorem{definition}{Definition}

\newtheorem{example}{Example}
\newtheorem{remark}{Remark}
\newtheorem{assumption}{Assumption}
\newtheorem{lemma}{Lemma}

\newtheorem{theorem}{Theorem}
\newtheorem{corollary}{Corollary}

\newtheorem{proposition}{Proposition}

\newtheorem{algorithm}{Algorithm}

\begin{document}

\title{\LARGE \bf Robust microphase separation through chemical reaction networks}


\author{Franco Blanchini$^{1}$, Elisa Franco$^{2}$, Giulia Giordano$^{3}$, Dino Osmanovi{\'{c}}$^{2}$
\thanks{This work has been partially supported by the U.S. N.S.F. (CAREER award 1938194 and FMRG:Bio award 2134772 to EF) and by the Sloan Foundation (award G-2021-16831). This work has also been partially funded by the European Union (NextGenerationEU grant Uniud-DM737 to FB and ERC INSPIRE grant 101076926 to GG); views and opinions expressed are however those of the authors only and do not necessarily reflect those of the European Union or the European Research Council Executive Agency, and neither the European Union nor the granting authority can be held responsible for them.
}
\thanks{$^{1}$ Department of Mathematics, Computer Science and Physics, University of Udine, Italy. \texttt{blanchini@uniud.it}}
\thanks{$^{2}$ Department of Mechanical and Aerospace Engineering, University of California at Los Angeles, USA.
\texttt{efranco@seas.ucla.edu, osmanovic.dino@gmail.com}}
\thanks{$^{3}$ Department of Industrial Engineering,
University of Trento, Italy. \texttt{giulia.giordano@unitn.it}}
}

\maketitle
\thispagestyle{empty} 
\begin{abstract}
The interaction of phase-separating systems with chemical reactions is of great interest in various contexts, from biology to material science. In biology, phase separation is thought to be the driving force behind the formation of biomolecular condensates, i.e. organelles without a membrane that are associated with cellular metabolism, stress response, and development. RNA, proteins, and small molecules participating in the formation of condensates are also involved in a variety of biochemical reactions: how do the chemical reaction dynamics influence the process of phase separation? Here we are interested in finding chemical reactions that can arrest the growth of condensates, generating stable spatial patterns of finite size (microphase separation), in contrast with the otherwise spontaneous (unstable) growth of condensates. We consider a classical continuum model for phase separation coupled to a chemical reaction network (CRN), and we seek conditions for the emergence of stable oscillations of the solution in space.  Given reaction dynamics with uncertain rate constants, but known structure, we derive easily computable conditions to assess whether microphase separation is impossible, possible for some parameter values, or robustly guaranteed for all parameter values within given bounds. Our results establish a framework to evaluate which classes of CRNs favor the emergence of condensates with finite size, a question that is broadly relevant to understanding and engineering life.
\end{abstract}


\section{Modelling Microphase Separation\\ in the Presence of Chemical Reactions}

\IEEEPARstart{P}{hase} separation has emerged as a key area within biological research over the last decade \cite{Banani2017}. By utilizing the physical properties of phase separation, it is hypothesized that living organisms are able to exercise fine control over their chemical production \cite{Klosin2020}. Within the non-equilibrium cellular environment, chemical reactions and phase separation combine to produce a new class of physical systems, deemed \textit{active emulsions} \cite{weber_2019} or \textit{active droplets} \cite{zwicker_2016}. Beyond biological relevance, these systems have displayed intriguing properties in their own right \cite{Li2020,Osmanovic2019}, yielding novel behavior in both spatial organization and dynamical properties. By combining both \textit{conserved} dynamics (phase separation) and \textit{non-conserved} dynamics (reactions), we can think of active emulsions as an extension of classical reaction-diffusion models \cite{arcac2011,Hori2015,Hori2019,Kashima2015,Miranda2021,murray2001mathematical}, and bring similar tools to bear on the analysis of their properties.


We model such systems by considering the time evolution of $n$ chemical species, with a vector of concentrations $c(z) =[c_1(z) \, c_2(z) \, \dots \, c_n(z)]\in \mathbb{R}^{n}_{+}$, $z\in \mathbb{R}^d$,
under the assumption that any of the species can undergo phase separation: it can be in two phases, condensed and dispersed, respectively characterized by concentrations $c_i^c$ and $c_i^d$.

\begin{figure}[b]
\centering 
\includegraphics[width=.8\columnwidth]{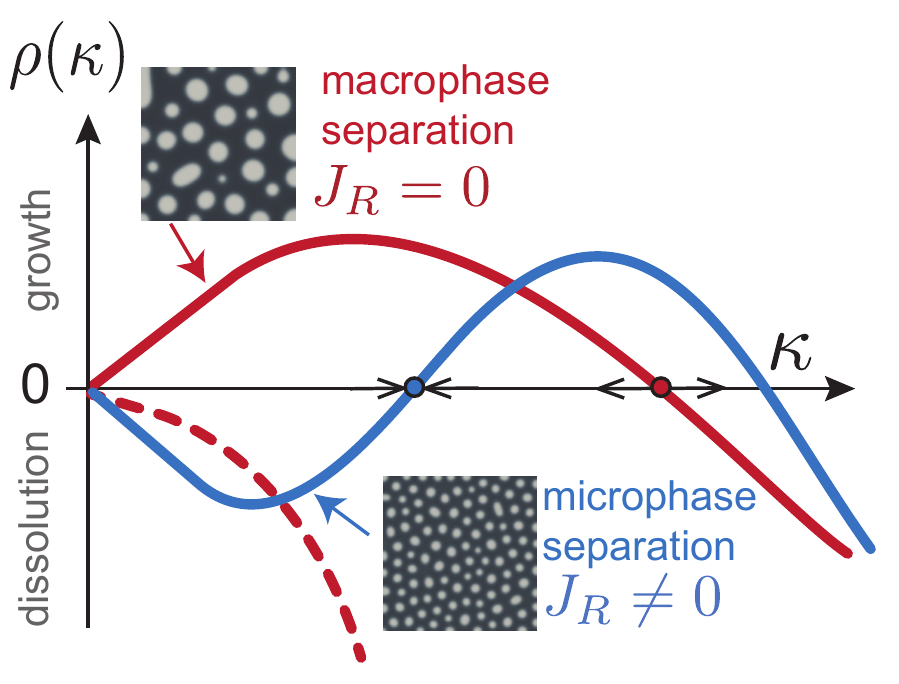}
\caption{Dispersion relation curves for phase separating systems. Condensate size grows when the wave number $\kappa$ decreases. When there is no separation (red, dashed curve), condensates dissolve. In macrophase separation (red, solid curve), typical in the absence of chemical reactions, condensates either grow until they are macroscopically separated (for small $\kappa$) or dissolve completely (for large $\kappa$), as the intermediate crossing point is an unstable fixed point. In microphase separation (blue curve), which can be induced by chemical dynamics, both large condensates (small $\kappa$) and small condensates (large $\kappa $) have a negative growth rate, while condensates of intermediate size have a positive growth rate, hence the mid crossing point is a stable fixed point, leading to a prevalent condensate size.}
\label{Fig:separation}
\end{figure}

The overall dynamics is  described by equation
\begin{equation} \label{eq:fum}
\frac{\mathrm d c(z,t)}{\mathrm d t} = \mathbf I(c(z,t))+ \mathbf R(c(z,t)),
\end{equation}
where the term
$\mathbf I(c(z,t))= \nabla \cdot \left(\mathfrak{D} \nabla\frac{\delta F(c(z,t))}{\delta c(z,t)}\right)$
describes conserved spatial dynamics by considering an energy functional $F(c(z,t))$ (see~\cite{osmanovic2022chemical}) and inferring the corresponding time evolution via ``model B'' dynamics~\cite{Hohenberg1977}. We assume that every species has the same homogeneous diffusion coefficient and molecular mass, hence the diffusion matrix is $\mathfrak{D}=d I$, where $d>0$ is a common diffusion constant. 
The term $\mathbf R(c(z,t))$ describes the reaction fluxes generated by a set of chemical reactions, assuming the existence of a free energy source that maintains the rates at which the reactions proceed~\cite{Osmanovic2019}.

Full analysis of \eqref{eq:fum} is usually rather complex, however we can obtain information about the properties of the solution by linearization, which can give tractable results on the effects of chemistry on phase separation. We perform a linear stability analysis of model~\eqref{eq:fum} near equilibrium: $c(z,t) =  c_s +  w \exp(i \kappa z + \rho(\kappa) t)$, where $w$ is small: $w^2\approx0$. The linearization around $c_s$ involves the Jacobian matrices $J_I(\kappa)=\nabla_{c(z)} \mathbf I(c( z,t))\big|_{c=c_s}$ of the conserved dynamics and $J_R=\nabla_{c(z)} \mathbf R(c(z,t))\big|_{c=c_s}$ of the chemical reaction dynamics. We assume that the phase-separating species is the first one, so that $J_I(\kappa)$ has the symmetric structure
\begin{equation}\label{peculiar}
J_I(\kappa)=\begin{bmatrix}
 \mu |\kappa|^2 -\gamma^2 |\kappa|^4 & -d|\kappa|^2 \epsilon_{12} & \ldots &  -d|\kappa|^2 \epsilon_{1N} \\
 -d|\kappa|^2 \epsilon_{12} & -d|\kappa|^2 &  \ldots & -d|\kappa|^2 \epsilon_{2N} \\
\vdots                         &  \vdots                       &    \ddots               & \vdots \\  
 -d|\kappa|^2 \epsilon_{1N} & -d|\kappa|^2 \epsilon_{2N} & \ldots & -d|\kappa|^2
\end{bmatrix},
\end{equation}
where $\mu>0$, $\gamma$ is the surface tension and $\epsilon_{ij}$ are constant parameters representing spatial attraction or repulsion among species. 
We thus obtain the relationship
\begin{equation}\label{DispersionRelation}
    \rho(\kappa)  w=  [J_I(\kappa)+J_R]  w,
\end{equation}
where the spectral abscissa $\rho(\kappa)$ of $(J_I(\kappa)+J_R)$ characterizes the dynamics of spatially oscillatory behaviors. The \emph{dispersion relation} curve $\rho(\kappa)$ depends on the wave number $\kappa$: the growth rate of a spatial wave depends on its wave number.

Linearizing the problem allows us to make quantitative predictions about the behavior of a given system without going through the expensive computation of the full solution to \eqref{eq:fum}. In particular, we can distinguish between systems that undergo \textit{microphase separation} (MS), \textit{macrophase separation} or \textit{no phase separation}. The first case corresponds to finite size patterns do not change over time: as illustrated in Fig.~\ref{Fig:separation}, this happens when $\rho( \kappa ) < 0$ for small $ \kappa  $, $\rho( \kappa) > 0$ for intermediate $\kappa $, and $\rho( \kappa ) < 0$ for large $\kappa $. It is of special interest as droplet size is regulated through the action of chemistry. Thus, spatial compartmentalization occurs with a particular length scale, a prerequisite to being able to utilize compartmentalization for precision control of spatially separated chemical reactions, creating droplets with ``life-like'' properties \cite{weber_2019}.

 
The design space of such a problem is large, in that we have many possible chemical reaction networks (CRNs) that can couple to phase separating systems. To assess the likelihood of MS and identify interesting candidate CRNs for experimental realization, in previous work we computationally explored the parameters and chemical reaction networks that lead to matrices $J_I(\kappa)$ and $J_R$ in \eqref{DispersionRelation}~\cite{osmanovic2022chemical}. The probability that $\rho(\kappa)$ has three roots was evaluated when generating either random $J_R$ matrices or random CRNs~\cite{VanDerSchaft2013}. 

In this paper, we re-examine the problem of inducing MS via chemical reactions through a control-theoretic approach. 
Given a CRN structure, whose rate parameters are unknown but bounded in a known range, we rely on parametric robustness approaches \cite{Barmish1994} and vertex results \cite{BlaColGioZor2020,BlaColGioZor2022,Giordano2016} to provide conditions ensuring that MS is: \textit{impossible}; \textit{possible for some parameter values}; or \textit{robustly guaranteed for all parameter values in the range}.

We achieve these conditions by converting the original problem, formulated in terms of the spectral abscissa of an uncertain matrix and thus challenging to handle, into a problem formulated in terms of the robust analysis of the determinant of an uncertain matrix, much simpler to deal with.

\section{Microphase Separation: a Spectral Problem}

To characterize the dispersion relation, we need to study the spectrum of matrix $J_I(\kappa)+J_R$. In view of \eqref{peculiar}, we can write $J_I(\kappa)= J_2 |\kappa|^2 + J_4 |\kappa|^4$ and the combined Jacobian in \eqref{DispersionRelation} becomes:
\begin{equation}\label{linBDC}
J(\kappa)=J_R + J_2 |\kappa|^2 + J_4 |\kappa|^4.
\end{equation}

Without restriction, we assume that $J_R$ can be rewritten according to the \emph{$BDC$ decomposition} introduced in~ \cite{BlanchiniGiordano2014,BlanchiniGiordano2021,Giordano2016}; this is possible for the Jacobian of any generic CRN.

\begin{assumption} The Jacobian $J_R$ can be decomposed as $J_R=B \Delta C$, where $\Delta=\mbox{diag}\{\Delta_1,
\Delta_2, \dots, \Delta_m\}$ has positive diagonal entries representing the uncertain parameters (the nonzero partial derivatives of the CRN system), while
matrices $B \in \mathbb{Z}^{n \times m}$ and  $C \in \mathbb{Z}^{m \times n}$ represent the known structure of the given CRN.
The unknown parameters $\Delta_j$ are bounded as
\begin{equation}\label{eq:Dbounds}
\Delta \in \mathcal{D} = \{\Delta \colon 0 \leq \Delta_j^- \leq \Delta_j\leq \Delta_j^+\},
\end{equation}
for given lower bounds $\Delta_j^-$ and upper bounds $\Delta_j^+$.
\end{assumption}

We restrict our analysis to CRNs with one conservation law.
\begin{assumption}\label{zero_eig}
For all $\Delta \in \mathcal{D}$, matrix $J_R = B\Delta C$ is singular and has $n-1$  eigenvalues with negative real part. 
Also, there exists a nonnegative vector $v^\top \geq 0$ such that
$v^\top B=0$, representing a conservation law: since $v^\top B\Delta C=0$, $v^\top$ is a left eigenvector of $J_R$ associated with the eigenvalue at $0$. 
\end{assumption}

Assumption~\ref{zero_eig} entails that matrix $J_R$ is marginally stable. Its stability can be assessed through several existing techniques tailored to CRNs, proposed for instance in~\cite{AlrAng16,BlanchiniGiordano2014,BCCG19,Clarke80}.

We also introduce suitable assumptions on the symmetric matrix $J_I(\kappa)=J_2 |\kappa|^2 + J_4 |\kappa|^4$.
\begin{assumption}\label{J2J4}
The symmetric matrix $J_2$ is indefinite, the symmetric matrix $J_4$ is negative semi-definite,
and there exists $\bar \kappa$ such that 
$J_I(\bar \kappa)= J_2 |\bar \kappa|^2 + J_4 |\bar \kappa|^4$ is negative definite.
\end{assumption}

The Jacobian matrix $J(\kappa)=J_R + J_2 |\kappa|^2 + J_4 |\kappa|^4$ in \eqref{linBDC} has eigenvalues $\lambda_i(\Delta,\kappa)$, $i=1,\dots,n$, and its spectral abscissa (namely, the maximum real part $\Re$ of its eigenvalues) is
\begin{equation*}    
\rho(\Delta,\kappa) = \max_i \left \{\Re  ( \lambda_i(\Delta,\kappa) \right \}.
\end{equation*}
The eigenvalues of $J_R=B\Delta C$ are $\lambda_i(\Delta,0)$, $i=1,\dots,n$.

\begin{assumption}
The eigenvalue of matrix $J_R=B\Delta C$ associated with the conservation law is $\lambda_1(\Delta,0)=0$.
\end{assumption}

\begin{definition}\label{SignChanges}
A continuous function $f(\kappa)$
has a {\em positive sign change} if $f(\kappa_1) <0<f(\kappa_2)$ for some $\kappa_1 < \kappa_2$, while it has a {\em negative sign change} if $f(\kappa_1) >0>f(\kappa_2)$. 
Moreover, function $f(\kappa)$ is {\em initially positive (respectively, negative)} if
there exists an open right neighborhood of $0$, $(0,\hat \kappa)$, in which the function is positive (respectively, negative).
\end{definition}

We can now define the (robust) microphase separation property that is the subject of our analysis.
\begin{definition}[Microphase separating system]\label{separation_condition}
System~\eqref{eq:fum} exhibits \emph{microphase separation (MS)} if, for a given $\Delta\in \mathcal{D}$, $\rho(\Delta,\kappa)$ is initially negative (i.e., there exists $\hat \kappa$ such that $\rho(\Delta,\kappa)<0$ for all $\kappa \in (0,\hat \kappa)$), then has a positive sign change 
and finally a negative sign change (i.e., $\rho(\Delta,\kappa_1)>0$ and $\rho(\Delta,\kappa_2)<0$ for some $0 < \hat \kappa < \kappa_1<\kappa_2$).
System~\eqref{eq:fum} exhibits \emph{robust MS} if this condition holds for all $\Delta \in \mathcal{D}$.
\end{definition}

The MS condition in Definition~\ref{separation_condition} describes a qualitative behavior of $\rho(\Delta,\kappa)$ consistent with the blue curve in Fig.~\ref{Fig:separation}: the size of condensates shrinks for small $\kappa$, grows for intermediate values of $\kappa$, shrinks again for large $\kappa$. For known parameters, the condition can be tested by directly computing the eigenvalue curves~\cite{osmanovic2022chemical}. Departing from this approach, our analysis aims to develop efficient methods to tackle the case of \textit{uncertain} CRN parameters: given a range of possible parameter values, by taking advantage of the $BDC$ decomposition of the Jacobian of the chemical reaction dynamics, we check whether the condition can or cannot hold for some parameters in the range, and whether it holds \textit{robustly} for all parameters in the range.

\section{Robust Determinant Conditions\\ for Microphase Separation}
We obtain (robust) conditions for MS by mapping the dispersion relation problem, involving the spectral abscissa $\rho(\Delta,\kappa)$, to a determinant problem.
To this aim, given uncertain CRN parameters $\Delta$, whose values are bounded in the set $\mathcal{D}$ as in \eqref{eq:Dbounds}, we consider the functions
 \begin{eqnarray}
\Psi^-(\kappa) &=& \min_{\Delta \in \mathcal{D}} \det [-(B\Delta C + |\kappa|^2 J_2 + |\kappa|^4 J_4) ],
\label{Psi-}
\\
\Psi^+(\kappa) &=& \max_{\Delta \in \mathcal{D}} \det [-(B\Delta C + |\kappa|^2 J_2 + |\kappa|^4 J_4) ],\label{Psi+}
 \end{eqnarray}
which can be easily computed, as we will show in Section~\ref{Sec:ComputePsi}.
In view of their definition,  $\Psi^-(\kappa)\leq \Psi^+(\kappa)$. Moreover, $\Psi^-(0)=\Psi^+(0)=0$,
because, in view of Assumption~\ref{zero_eig}, the determinants in \eqref{Psi-} and \eqref{Psi+} are $0$ for $\kappa=0$.
Both functions grow to infinity as $\kappa \to +\infty$, as shown after Lemma~\ref{tend_to_minus_inf}.

Studying functions $\Psi^-(\kappa)$ and $\Psi^+(\kappa)$ allows us to provide:
\begin{itemize}
\item a crucial necessary condition for MS: 
$\Psi^+$ needs to be initially positive (Theorem~\ref{Th:possible}); 
\item a sufficient condition ensuring MS for
some values of the parameters (which we can determine):
{\em either} $\Psi^-$ {\em or} $\Psi^+$ is initially positive and has a negative sign change (Theorem~\ref{compatible}); 
\item a sufficient condition ensuring \textit{robust} MS for all admissible parameters: {\em both} $\Psi^-$ and $\Psi^+$ are initially positive and have a negative sign change (Theorem~\ref{Th:robust}).
\end{itemize}

The technical challenge lies in relating
the spectral abscissa $\rho(\Delta,\kappa)$ to the curves $\Psi^-$ and $\Psi^+$.
The difficulty arises from the fact that, while a negative value
of the determinant $\det [-(B\Delta  C + |\kappa|^2 J_2 + |\kappa|^4 J_4) ]$ implies positivity of the spectral abscissa $\rho(\Delta,\kappa)$, the opposite unfortunately is not true: the determinant may well be positive even when $\rho(\Delta,\kappa)>0$.

We begin by considering the case in which only the chemical reaction parameters are subject to uncertainty, which affects the entries of $J_R$.
 \begin{theorem}[Necessary condition]\label{Th:possible}
If system~\eqref{eq:fum} exhibits microphase separation for some $\Delta \in \mathcal{D}$, then $\Psi^+$ is initially positive.
\end{theorem}
\begin{proof}
Since $ \rho(\Delta,\kappa)<0$ necessarily requires that
$\det [-(B\Delta  C + |\kappa|^2 J_2 + |\kappa|^4 J_4) ] >0$,
then the maximum $\Psi^+$ must be initially positive when $ \rho(\Delta,\kappa)$ is initially negative, as required by Definition~\ref{separation_condition}.
\end{proof}

In Fig.~\ref{Fig2}A we illustrate a case in which $\Psi^+$ is not initially positive, hence MS is not possible; conversely, in Fig.~\ref{Fig2}B, $\Psi^+$ is initially positive.

We now state the other main results, whose proof requires some technical lemmas and is thus reported in Section~\ref{sec:proofs}.
 
\begin{theorem}[Sufficient condition]\label{compatible} 
System~\eqref{eq:fum} exhibits microphase separation for some $\Delta \in \mathcal{D}$ if either $\Psi^-$ or $\Psi^+$ is initially positive and has a negative sign change.
\end{theorem}

Actual parameter values for which MS does occur can be found following the procedure described in  Section~\ref{Sec:ComputePsi}.

In Figs.~\ref{Fig2}B and \ref{Fig2}C, we illustrate the case in which MS is possible for some $\Delta \in \mathcal{D}$.

\begin{theorem}[Robust sufficient condition]\label{Th:robust}
System~\eqref{eq:fum} exhibits microphase separation for all  $\Delta \in \mathcal{D}$ if $\Psi^-$ is initially positive (and hence $\Psi^+$ is initially positive too) and $\Psi^+$ has a negative sign change (and hence $\Psi^-$ has it too).
\end{theorem}

In Fig.~\ref{Fig2}D, we illustrate the case in which MS is robustly guaranteed for all $\Delta \in \mathcal{D}$.

Finally, we briefly touch upon the case in which uncertainty affects the parameters of both spatial dynamics (i.e., the entries of $J_I$) and chemical dynamics (i.e., the entries of $J_R$).

\begin{corollary}\label{Cor:diagJ}
Assume that $J_2$ and $J_4$ are diagonal matrices, whose diagonal entries are uncertain parameters bounded in intervals. Then, redefining the functions  as 
$\Psi^-(\kappa) = \min_{\Delta,J_2,J_4} \det [-(B\Delta C + |\kappa|^2 J_2 + |\kappa|^4 J_4) ]$ and
$\Psi^+(\kappa) = \max_{\Delta,J_2,J_4} \det [-(B\Delta C + |\kappa|^2 J_2 + |\kappa|^4 J_4) ]$,
all the previous results still hold true.
\end{corollary}

\subsection{Proofs of the main results}\label{sec:proofs}
We begin with two technical lemmas. The first states that, for large enough $\kappa$, we have 
Hurwitz stability.
\begin{lemma}\label{tend_to_minus_inf}
For any given $\Delta$,
$\lim_{\kappa\to+\infty} \rho(\Delta,\kappa) = -\infty$.
\end{lemma}
\begin{proof} Recall that $J_4$ is negative semi-definite and take $\kappa \geq \bar \kappa$, with $\bar \kappa$ defined in Assumption~\ref{J2J4}. The Lyapunov inequality
\begin{eqnarray}
 &&\left( J_R + J_2 |\kappa|^2 + J_4 |\kappa|^4 \right)^\top +  
 \left( J_R + J_2 |\kappa|^2 + J_4 |\kappa|^4 \right)= \nonumber\\
 && (B\Delta C)^\top +B\Delta C + 2 J_2 |\kappa|^2 +2 J_4 |\kappa|^4 \leq \nonumber\\
 && (B\Delta C)^\top +B\Delta C + 2 J_2 |\kappa|^2 +2 J_4 |\kappa|^2  |\bar \kappa|^2 =\nonumber\\
 && |\kappa|^2 \left [\frac{(B\Delta C)^\top +B\Delta C}{|\kappa|^2} + 2  ( J_2   + J_4 |\bar \kappa|^2) \right ] <0 \label{lequa}
 \end{eqnarray}
holds for $\kappa$ large, because $( J_2   + J_4 |\bar \kappa|^2)$ is negative definite and
 the fraction converges to $0$. This implies Hurwitz stability for $\kappa$ large enough.
 To prove that $\lim_{\kappa\to\infty}\rho(\Delta,\kappa) = -\infty$,
denote $K = B\Delta C+\kappa I$ and repeat the computation for the perturbed matrix $J_\kappa =\left( \kappa I + J_R + J_2 |\kappa|^2 + J_4 |\kappa|^4 \right)$ to get
 \begin{equation}
J_\kappa^\top +  
 J_\kappa \leq |\kappa|^2 \left [\frac{K^\top +K}{|\kappa|^2} + 2  ( J_2   + J_4 |\bar \kappa|^2) \right ] <0
 \end{equation}
 for $\kappa$ large, hence $J_\kappa$ is Hurwitz. The spectral abscissa of $\left(J_R + J_2 |\kappa|^2 + J_4 |\kappa|^4 \right)$ is thus less than $-\kappa$, for $\kappa$ large.
\end{proof}

Since the determinant of the negative of a matrix of size $n$ is $(-1)^n$ times the product of the matrix eigenvalues, $\lim_{\kappa\to+\infty} \rho(\Delta,\kappa) = -\infty$ implies that $\lim_{\kappa\to+\infty} \Psi^-(\kappa)=\lim_{\kappa\to+\infty} \Psi^+(\kappa)= +\infty$.

\begin{remark}
If we take bounds $\Delta^-_i >0$ and $\Delta^+_i< \infty$, by compactness the limit in Lemma~\ref{tend_to_minus_inf} is uniform in $\Delta \in \mathcal{D}$.
\end{remark}

Due to Lemma~\ref{tend_to_minus_inf}, if $\rho(\Delta,\kappa)$ has a positive sign change, then it needs to have a subsequent negative sign change.

\begin{lemma}\label{key}
Under Assumption~\ref{zero_eig}, if for some $\Delta$ we have that
$\det [-(B\Delta  C + |\kappa|^2 J_2 + |\kappa|^4 J_4) ]$
is initially positive, as a function of $\kappa$,
  then there exists $\tilde \kappa>0$ 
such that $\rho(\Delta,\kappa)<0$ for $0<\kappa \leq \tilde \kappa$.
\end{lemma}
\begin{proof}
Consider the characteristic polynomial
\begin{equation}\label{eq:charpoly}    
p(s,\kappa,\Delta) = 
\det [sI-(B\Delta  C + |\kappa|^2 J_2 + |\kappa|^4 J_4) ].
\end{equation}
For $\kappa=0$, $\lambda_1=0$ is an isolated root of the characteristic polynomial $p(s,0,\Delta)$, while all the other roots $\lambda_i$, $i >1$, have negative real parts. In view of the continuity of the eigenvalues with respect to $\kappa$,
for a small, positive $\kappa$ the roots 
$\lambda_i$, $i >1$, still have negative real part.
Hence, the spectral abscissa is given by the dominant real eigenvalue $\lambda_1$ and
we show that, for a small, positive $\kappa$, $\lambda_1$ becomes negative, hence $\rho(\Delta,\kappa)<0$.
We can write the characteristic polynomial as
$$p(s,\kappa,\Delta) = s^n + p_{n-1}(\kappa,\Delta) s^{n-1} + \dots + p_1 (\kappa,\Delta) s+p_0 (\kappa,\Delta).
$$
For $\kappa=0$, the constant term $p_0 (0,\Delta)=0$, due to the zero root, while
$p_j (0,\Delta)>0$ for $j>0$, since all other roots have negative real part in view of Assumption~\ref{zero_eig}.
Hence the derivative of $p$ computed at $s=0$, assuming $s$ real, is positive:
$\left. \frac{d}{ds}p(s,\kappa,\Delta)\right|_{s=0} =  p_1 (\kappa,\Delta) > 0$.
For a small, positive $\kappa$ in a neighborhood of $0$, the polynomial $p(s,\kappa,\Delta)$
becomes positive and, since it is locally increasing in $\kappa$,
the root $\lambda_1$, which is initially $0$,  moves to the left and becomes negative.
\end{proof}

We are now ready to prove Theorem~\ref{compatible}.

\begin{proof}
If $\Psi^-$ is initially positive,
then for all $\Delta$, $p(s,\kappa,\Delta)$ defined in \eqref{eq:charpoly} is initially positive.
Lemma~\ref{key} ensures that $ \rho(\Delta,\kappa)<0$ in a right neighborhood of zero,
for all $\Delta$.
On the other hand, $\Psi^-$ becomes negative for some larger $\kappa$, meaning
that  for some $\Delta^*$, $\rho(\Delta^*,\kappa)>0$, so 
$\rho(\Delta^*,\kappa)$ has a positive sign change.
Then, $\rho$ will have
a negative sign change $\rho(\Delta^*,\kappa)<0$ for $\kappa$ large
in view of Lemma~\ref{tend_to_minus_inf}.

If $\Psi^+$ is initially positive,
then for some $\Delta^*$, $p(s,\kappa,\Delta^*)$ is initially positive.
In view of Lemma~\ref{key}, $ \rho(\Delta^*,\kappa)<0$ in a right neighborhood of zero,
for all $\Delta$.
On the other hand, $\Psi^+$ becomes negative for some larger $\kappa$, meaning
that  $\rho(\Delta,\kappa)>0$ for all $\Delta$. Then $\rho(\Delta^*,\kappa)<0$
has a positive sign change, and eventually  will have a negative sign change $\rho(\Delta^*,\kappa)<0$ for $\kappa$ large, again,
in view of Lemma~\ref{tend_to_minus_inf}.
\end{proof}

Finally, we prove Theorem~\ref{Th:robust}.

\begin{proof}
If  $\Psi^-$ is initially positive, then $\rho(\Delta,\kappa)$ is initially negative for all $\Delta$. If $\Psi^+$ has a negative sign change, then  $\rho(\Delta,\kappa)$ becomes positive for all $\Delta$ 
(before becoming eventually negative).
\end{proof}

\begin{remark}
The value $\kappa^*$ at which $\rho(\Delta,\kappa)$ becomes positive
(transition to instability) clearly depends on $\Delta$. 
In general, it is not possible to discriminate whether the instability occurs due to the appearance of real or complex eigenvalues, unless $J_R$ is Metzler, as in the case of a network of mono-molecular reactions, and $J_I$ is diagonal, so that the dominant eigenvalue is real. In general, a supplementary analysis can be performed adopting value-set techniques \cite{Barmish1994} to possibly rule out imaginary eigenvalues.
\end{remark}

Corollary~\ref{Cor:diagJ} can be proven by repeating the same arguments as above, since 
the determinant is a multilinear function of all the considered parameters if  
matrices $J_2$ and $J_4$ are diagonal.


\subsection{Computing $\Psi^-$ and $\Psi^+$ and checking criteria}\label{Sec:ComputePsi}
Computing $\Psi^-$ and $\Psi^+$ is simple,
since they are respectively the minimum and the maximum of $2^m$ polynomials.
\begin{proposition}\label{prop:vertex}
Given the set ${\mathcal{D}}$ in \eqref{eq:Dbounds}, let $\hat{\mathcal{D}}$ be the set of all its vertices:
$$
\hat{\mathcal{D}} = \left \{\Delta: \Delta_{j} \in \{\Delta_{j}^-, \Delta_{j}^+\} \right \}.
$$
Then, the functions $\Psi^-$ and $\Psi^+$, respectively defined in \eqref{Psi-} and \eqref{Psi+}, can be computed as
\begin{equation}    
\Psi^-(\kappa)  = \min_{i} p_i(\kappa), \qquad \Psi^+(\kappa) = \max_{i} p_i(\kappa),
\end{equation}
where
\begin{equation*}
p_i(\kappa) = \det [-(B\Delta^{(i)}  C + |\kappa|^2 J_2 + |\kappa|^4 J_4) ]
\end{equation*}
with $\Delta^{(i)}  \in \hat{\mathcal{D}}$, $i=1,2,\dots,2^m$.

The curves $\Psi^-$ and $\Psi^+$ are thus piecewise polynomials: there exists a finite set of values $\kappa_1,\kappa_2,\dots,\kappa_M$
for which $\Psi^-$ (analogously, $\Psi^+$) is a polynomial in each interval $[\kappa_h,\kappa_{h+1}]$. Hence, $\Psi^-$ and $\Psi^+$ are piecewise differentiable.
\end{proposition}
\begin{proof}
For any value of $\kappa$, the maximum and minimum in \eqref{Psi-} and \eqref{Psi+} are achieved on the vertices $\hat{\mathcal{D}}$
of the hyper-rectangle ${\mathcal{D}}$, because the determinants are multilinear functions of the uncertain parameters $\Delta_i$ and any multilinear function defined on a hyper-rectangle achieves both its minimum and its maximum value on a vertex \cite{Barmish1994,BlaColGioZor2020,BlaColGioZor2022}.

The fact that the curves $\Psi^-$ and $\Psi^+$ are piecewise polynomials follows from the fact that two different polynomials of order $N$ can intersect in at most $N$ points. 
\end{proof}

It can be immediately seen that Proposition~\ref{prop:vertex} allows us to check the conditions of Theorems~\ref{Th:possible}, \ref{compatible} and \ref{Th:robust} as follows.
 \begin{proposition} The following equivalences hold.
\begin{itemize}
\item $\Psi^-$ is initially positive iff all the polynomials $p_i(\kappa)$ are initially positive.
\item $\Psi^+$ is initially positive iff at least one of the polynomials $p_i(\kappa)$ is initially positive.
\item $\Psi^-$ takes negative values iff at least one of the polynomials $p_i(\kappa)$ does.
\item $\Psi^+$ takes negative values iff all the polynomials $p_i(\kappa)$ do.
\end{itemize}
\end{proposition}

Remarkably, checking the conditions for MS requires the 
analysis of a finite number of polynomials.

When the sufficient conditions of Theorem~\ref{compatible} are met, and thus the considered CRN structure can give rise to MS \emph{provided that the parameters are suitably chosen}, the following algorithm allows us to identify values that do ensure MS.

\begin{algorithm}\label{values_algorithm}
To find parameter values ensuring MS:
\begin{itemize}
\item If $\Psi^-$ is initially positive and has a negative sign change,
find $\kappa^*$ for which $\Psi^-(\kappa^*)<0$.
Then, the vertex polynomial $p_h$ such that $p_h(\kappa^*)=\Psi^-(\kappa^*)<0$ is associated with vertex parameters $\Delta_{sep}^* \in \hat{\mathcal{D}}$ ensuring MS.
\item If $\Psi^+$ is initially positive and has a negative sign change,
find $\tilde \kappa$ such that, for $0 < \kappa \leq \tilde \kappa$, the vertex polynomial $p_h(\kappa)=\Psi^+(\kappa)>0$.
Then, $p_h$ is associated with vertex parameters $\tilde  \Delta_{sep} \in \hat{\mathcal{D}}$
ensuring MS.
\end{itemize}
\end{algorithm}

\section{Testing Microphase Separation in the Presence of Chemical Reactions}

We provide here a collection of examples where we apply our proposed criteria to test robust microphase separation in the presence of different chemical reaction network structures. In all the considered examples, we assume
$|\kappa|^2 J_2 + |\kappa|^4 J_4 = \mbox{diag}\begin{bmatrix} \mu |\kappa|^2  - \gamma^2 |\kappa|^4  & -d|\kappa|^2  & -d|\kappa|^2 & \dots & -d|\kappa|^2 \end{bmatrix}$,
so that the separating species is the first one ($C_1$), and we take $\mu=1$, $d=1$ and $\gamma^2=0.03$.

\begin{example}\label{ex:impossible}
Consider the CRN
$C_1 + C_3 \rightharpoonup C_2 + C_4$, $C_2   \rightharpoonup  C_1$, $C_4  \rightharpoonup \emptyset$, $\emptyset \rightharpoonup C_3$.
The CRN system is formed by equations $\dot c_1 = -g_{13}(c_1,c_3) + g_2(c_2) = -\dot c_2$, $\dot c_3 = -g_{13}(c_1,c_3) + c_0$, $\dot c_4 = +g_{13}(c_1,c_3) - g_4(c_4)$, corresponding to the $BDC$-decomposable Jacobian matrix
\begin{equation*}
J=\underbrace{\begin{bmatrix}
     -1  &   1  &  -1  &   0\\
     1  &  -1  &   1   &  0\\
    -1   &  0  &  -1  &   0\\
     1   &  0  &   1  &  -1
     \end{bmatrix}}_{B} ~~\Delta ~~
     \underbrace{\begin{bmatrix}
      1   &  0  &   0  &   0\\
     0   &  1  &   0  &   0\\
     0  &   0    & 1  &   0\\
     0  &   0  &   0  &   1
     \end{bmatrix}}_{C},
\end{equation*}
where
$\Delta=\mbox{diag} \begin{bmatrix}  \frac{\partial g_{13}}{\partial c_1}   &  \frac{\partial g_{2}}{\partial c_2}   &  \frac{\partial g_{13}}{\partial c_3}    &   \frac{\partial g_4}{\partial c_4}\end{bmatrix}$.
With bounds $\Delta_i^+=0.5$ and $\Delta_i^-=0.3$ for all $i$, functions $\Psi^-$ and $\Psi^+$ are visualised in Fig. \ref{Fig2}A. The necessary condition in Theorem~\ref{Th:possible} is violated, because $\Psi^+$ is initially negative, hence MS is never possible for $\Delta \in \mathcal{D}$.

\begin{figure*}[ht]
\centering 
\includegraphics[width=\textwidth]{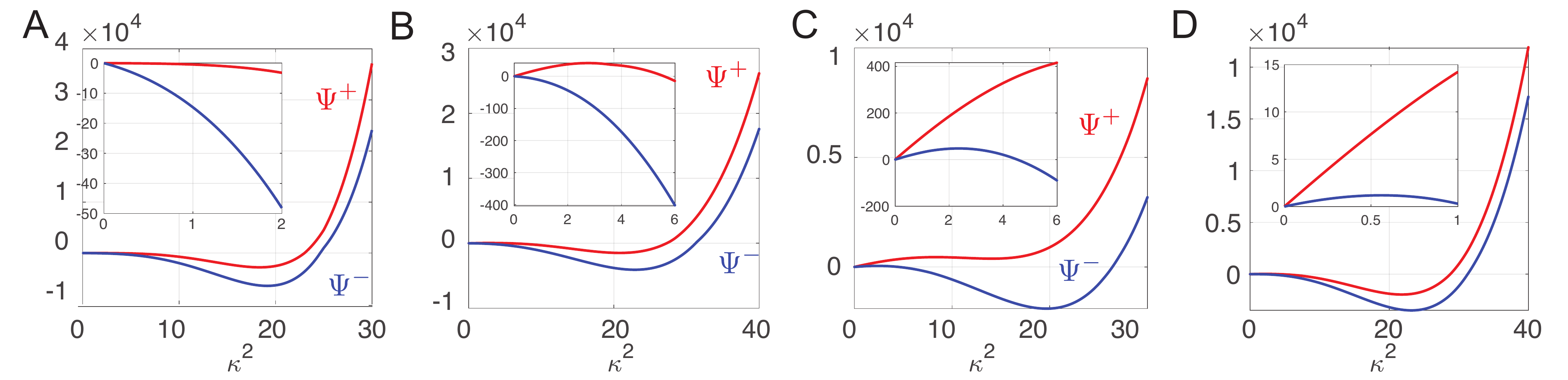}
\vspace{-.5cm}
\caption{$\Psi^-$ (blue) and $\Psi^+$ (red) for the different examples we considered. In Example~\ref{ex:impossible} (panel A), $\Psi^+$ is initially negative, thus violating the necessary condition in Theorem~\ref{Th:possible}: MS never occurs for $\Delta \in \mathcal{D}$. In Example~\ref{ex:possible1} (panel B), $\Psi^+$ is initially positive and has a negative sign change: MS occurs for some  $\Delta \in \mathcal{D}$. In Example~\ref{ex:possible2} (panel C), $\Psi^-$ is initially positive and has a negative sign change: MS occurs for some  $\Delta \in \mathcal{D}$. In Example~\ref{ex:robust} (panel D), $\Psi^-$ (hence $\Psi^+$) is initially positive and then $\Psi^+$ (hence $\Psi^-$) changes sign: MS occurs robustly for all $\Delta \in \mathcal{D}$.}
\vspace{-.5cm}
\label{Fig2}
\end{figure*}
\end{example}


\begin{example}\label{ex:possible1}
Consider the CRN
$C_1  \rightharpoonup   C_2   \rightharpoonup  C_3  \rightharpoonup C_1$
corresponding to equations
$\dot c_1= -g_{1}(c_1) + g_3(c_3)$, $\dot c_2= -g_{2}(c_2) + g_1(c_1)$, $\dot c_3= -g_{3}(c_3) + g_2(c_2)$,
with Jacobian
\begin{equation*}
J
=\underbrace{\begin{bmatrix}
     -1  &   0   &  1\\
     1  &  -1    & 0\\
     0  &   1   & -1
     \end{bmatrix}}_{B} ~~\Delta ~~
     \underbrace{\begin{bmatrix} 
     1 &    0   &  0\\
     0  &   1  &   0\\
     0  &   0  &   1
     \end{bmatrix}}_{C},
\end{equation*}
where
$ \Delta=\mbox{diag} \begin{bmatrix}  \frac{\partial g_{1}}{\partial c_1}   &  \frac{\partial g_{2}}{\partial c_2}   &  \frac{\partial g_{3}}{\partial c_3}
     \end{bmatrix}$.
Taking $\Delta_i^-=2$ and $\Delta_i^+=5$ for all $i$, Fig.~\ref{Fig2}B shows that MS is possible, because the condition in Theorem~\ref{compatible} is satisfied: $\Psi^+$ is initially positive and has a negative sign change. A parameter choice ensuring MS is $\tilde \Delta_{sep} = \begin{bmatrix} 5   &  2 &    2\end{bmatrix}$, computed following Algorithm~\ref{values_algorithm}.

 \end{example}

\begin{example}\label{ex:possible2}
Consider the CRN
$C_1  \rightleftharpoons   C_2$, $C_1+C_2   \rightleftharpoons C_3$ and the corresponding system
$\dot c_1 = -g_{12}(c_1,c_2) - g_1(c_1) + g_2(c_2) + g_3(c_3)$, $\dot c_2 = -g_{12}(c_1,c_2) + g_1(c_1) - g_2(c_2) + g_3(c_3)$, $\dot c_3 = g_{12}(c_1,c_2) - g_3(c_3)$, which has Jacobian
\begin{equation*}
J=    
\underbrace{\begin{bmatrix}
    -1   &  1  &  -1   & -1 &    1\\
     1  &  -1  &  -1 &   -1 &    1\\
     0  &   0 &    1  &   1  &  -1
\end{bmatrix}}_{B}
\Delta~
\underbrace{\begin{bmatrix}
     1   &  0  &   0\\
     0   &  1 &    0\\
     1    & 0  &   0\\
     0   &  1  &   0\\
     0   &  0  &   1
     \end{bmatrix}}_{C},
\end{equation*}
with
$\Delta=\mbox{diag} \begin{bmatrix}  \frac{\partial g_{1}}{\partial c_1}   &  \frac{\partial g_{2}}{\partial c_2}   &  \frac{\partial g_{12}}{\partial c_1}    &   \frac{\partial g_{12}}{\partial c_2}   &
\frac{\partial g_{3}}{\partial c_3}
     \end{bmatrix}$.
When $\Delta_i^-=3$ and $\Delta_i^+=5$ for all $i$, as shown in Fig.~\ref{Fig2}C, MS is possible, because the condition in Theorem~\ref{compatible} is satisfied: $\Psi^-$ is initially positive and has a negative sign change.
A parameter choice ensuring MS is $\Delta^*_{sep} =\begin{bmatrix}   3  &   5   &  3  &   3  &   5 \end{bmatrix}$, computed following Algorithm~\ref{values_algorithm}.

 \end{example}

\begin{example}\label{ex:robust}
Consider the CRN
$ C_1  \rightleftharpoons   C_2   \rightleftharpoons  C_3$, corresponding to system
$\dot c_1= -g_{1}(c_1) + g_{2,1}(c_2)$, $\dot c_2=-g_{2,1}(c_2) + g_1(c_1) -g_{2,2}(c_2) +g_3(c_3)$, $\dot c_3=-g_{3}(c_3) + g_{2,2}(c_2)$,
which has Jacobian
\begin{equation*}    
J=\underbrace{\begin{bmatrix}
     -1  &   1 &   0     &  0\\
     1   & -1   & -1     &  1\\
     0   &  0   &  1      &-1   \end{bmatrix}}_{B}
\Delta ~
\underbrace{\begin{bmatrix}
     1   &  0  &   0\\
     0   &  1  &   0\\
     0   &  1  &   0\\
     0   &  0  &   1   \end{bmatrix}}_{C}
\end{equation*}
with
$ \Delta=\mbox{diag} \begin{bmatrix}  \frac{\partial g_{1}}{\partial c_1}   &  \frac{\partial g_{2,1}}{\partial c_2}   &  \frac{\partial g_{2,2}}{\partial c_2}   &\frac{\partial g_{3}}{\partial c_3}
     \end{bmatrix}$.
Taking $\Delta_i^-=1$ and $\Delta_i^+=2$  for all $i$,
Fig.~\ref{Fig2}D shows that MS occurs robustly, for all choices of $\Delta$ within the bounds, because the condition in Theorem~\ref{Th:robust} is satisfied: both $\Psi^-$ and $\Psi^+$ are initially positive and have a negative sign change.

 \end{example}
     
\section{Concluding Discussion}     
Reaction-Diffusion equations are fundamental in the modelling of biological systems across different scales, from concentrations of different molecular species to spatial dynamics of reproducing organisms. The common feature of all such models is the coupling of non-conserved dynamics, which modifies the total concentration of species (reactions), to conserved dynamics, which moves a given concentration in space (diffusion). Recent work \cite{Cates,Chan2019,Li2020,zwicker2022intertwined} has extended the reaction-diffusion framework to account for more complex models of diffusive dynamics, such as usage of the Cahn-Hilliard functional to model the conserved dynamics corresponding to phase separation. While traditional reaction-diffusion systems have been an intense object of study \cite{arcac2011,Hori2015,Hori2019,Kashima2015,Miranda2021,murray2001mathematical}, full treatment of phase separating systems subject to chemical reactions is only now beginning to be explored, with multiple recent studies pointing to intriguing structural and dynamical properties of such systems \cite{zwicker_2016}.

We have considered the problem of predicting the emergence of microphase separation (MS) in a continuum model that couples phase separation and chemical reactions: spatial dynamics affect how species arrange in space while keeping their total concentration constant, and chemical reaction dynamics determine how the species locally inter-convert, thus changing the total amounts of individual components. MS is associated with the occurrence stable spatial oscillations with regions of high density of material, known as condensates. We have considered uncertain chemical reaction parameters bounded in a known interval and provided easy-to-compute conditions to check whether MS can be ruled out, or can arise for some parameters in the interval, or does robustly arise for all parameters in the interval. Our conditions offer useful insight for the robust experimental design of phase-separating systems in synthetic biology or material science.

\bibliographystyle{IEEEtran}

\end{document}